\documentclass[%
reprint,
superscriptaddress,
 amsmath,amssymb,
 aps,
 pra,
longbibliography,
twocolumn
]{revtex4-2}

\usepackage{graphicx}
\usepackage{dcolumn}
\usepackage{bm}

\usepackage[utf8]{inputenc}
\usepackage[english]{babel}
\usepackage[T1]{fontenc}

\usepackage{tikz}
\usepackage{lipsum}
\usepackage[colorlinks,linkcolor=blue,citecolor=blue]{hyperref}

\usepackage{amsmath,amsfonts,amssymb,amsthm,ascmac}
\usepackage{mathtools}
\usepackage{mathrsfs}
\usepackage{siunitx}
\usepackage{bm}
\usepackage{cancel}

\usepackage{enumerate}
\usepackage{siunitx}
\usepackage{physics}
\usepackage{pxrubrica}
\usepackage[version=3]{mhchem}
\usepackage{array}
\usepackage{float}

\usepackage{booktabs}
\usepackage{multirow}
\usepackage{hhline}
\usepackage{subcaption}
\usepackage{graphicx}
\usepackage{tikz}
\usetikzlibrary{intersections,calc,patterns,through,positioning,arrows,backgrounds,cd}
\usepackage{pgfplots}
\usepgfplotslibrary{patchplots}
\usepackage{titlesec}
\usepackage{picture}
\usepackage{fancybox}
\usepackage{boites}
\usepackage{tcolorbox}
\usepackage{dsfont}
\usepackage{ulem}

\usepackage{algorithm}
\usepackage{algpseudocode}
\usepackage{listings}
\usepackage{mathtools}
\usepackage{amsthm}
\usepackage{xcolor}
\usepackage{fancybox}
\usepackage{stmaryrd}
\usepackage{ytableau}
\usepackage{dsfont}
\usepackage{physics}
\usepackage[justification=raggedright, singlelinecheck=true]{caption}
\usepackage{tabularx}
\usepackage{ragged2e}
\usepackage[english]{babel}

\usepackage{todonotes}

\ytableausetup{boxsize=1em}

\newtheorem{theorem}{Theorem}

\newtheorem{lemma}{Lemma}
\newtheorem{problem}{Problem}
\theoremstyle{remark}

\newcommand{\black}[1]{\textcolor{black}{#1}}

\newcommand{\mac}[1]{\mathcal{#1}}
\newcommand{\mtr}[1]{\mathrm{#1}}

\DeclareRobustCommand{\erase}{\bgroup\markoverwith{\textcolor{red}{\rule[.5ex]{2pt}{0.4pt}}}\ULon}
\newcommand{\coloneq}{:=}

\begin{document}


\title{\mbox{Faster Quantum Algorithm for Multiple Observables Estimation in Fermionic Problems}}

\author{Yuki Koizumi}
\email{koizumiyuki903@gmail.com}
 \affiliation{Department of Applied Physics, University of \mbox{Tokyo, 7-3-1 Hongo, Bunkyo-ku, Tokyo 113-8656, Japan}}
\author{Kaito Wada}%
\affiliation{\mbox{Graduate School of Science and Technology, Keio University, 3-14-1 Hiyoshi, Kohoku, Yokohama, Kanagawa, 223-8522, Japan}}

\author{Wataru Mizukami}
\affiliation{\mbox{Center for Quantum Information and Quantum Biology, Osaka
University, 1-2 Machikaneyama, Toyonaka, Osaka 560-0043, Japan}}
\affiliation{Graduate School of Engineering Science, \mbox{Osaka University, 1-3 Machikaneyama, Toyonaka, Osaka 560-8531, Japan}}

\author{Nobuyuki Yoshioka}
\email{ny.nobuyoshioka@gmail.com}
\affiliation{\mbox{International Center for Elementary Particle Physics, The University of Tokyo, 7-3-1 Hongo, Bunkyo-ku, Tokyo 113-0033, Japan}}




\begin{abstract}
Achieving quantum advantage in efficiently estimating collective properties of quantum many-body systems remains a fundamental goal in quantum computing. 
While the quantum gradient estimation (QGE) algorithm has been shown to achieve doubly quantum enhancement in the precision 
and the number of observables, it remains unclear whether one benefits in practical applications.
In this work, we present a generalized framework of adaptive QGE algorithm, and further propose two variants which enable us to estimate the collective properties of fermionic systems using the smallest cost among existing quantum algorithms. The first method utilizes the symmetry inherent in the target state, and the second method enables estimation in a single-shot manner using the parallel scheme. 
We show that our proposal offers a quadratic speedup compared with prior QGE algorithms in the task of fermionic partial tomography for systems with limited particle numbers.
Furthermore, we provide the numerical demonstration that, for a problem of estimating fermionic 2-RDMs, our proposals improve the number of queries to the target state preparation oracle by a factor of 100 for the nitrogenase FeMo cofactor and by a factor of 500 for Fermi-Hubbard model of 100 sites. 
\end{abstract}

\maketitle


\textit{Introduction.---}
Achieving quantum enhancement in extracting information from a complex quantum system is one of the central challenges in quantum information science. 
It was initially posed in the context of quantum sensing that there is a quadratic gap between ordinary statistical sampling and quantum mechanical limitation, known as the Heisenberg limit (HL)~\cite{giovannetti2004quantum, giovannetti2006metrology, holland1993interferometric, demkowicz2015quantum, brassard2002amplitude}.
Indeed, for single parameter estimation, the HL scaling has been realized experimentally via adaptive scheme of phase estimation algorithm~\cite{higgins2007entanglement, 2009NJPh...11g3023H}.
Beyond eigenphase of a given unitary, it has been found that quantum amplitude estimation (QAE) algorithm achieves the HL scaling for general expectation value estimation~\cite{brassard2002amplitude, rebentrost2018quantum, Obrien2022efficient, suzuki2020amplitude, rall2020estimating}.

It is natural to next ask whether quantum enhancement is feasible in gathering collective properties.
Leading candidates beyond straightforward application of QAE algorithm are based on the quantum gradient estimation (QGE) algorithm~\cite{jordan2005fast, gilyen2019gradient, huggins2022nearly, apeldoorn2023quantum, wada2024Heisenberg}, a multi-parameter extension of the phase estimation. 
QGE-based algorithms construct entanglement between the target system and probe system that collectively encodes the information of multiple observables, leading to doubly quantum enhancement; quadratic improvement not only regarding the measurement uncertainty but also the number of observables.
While initial proposals showed such a performance regarding the additive error~\cite{huggins2022nearly, apeldoorn2023quantum}, it was pointed out to fall short in the worst case scenario~\cite{wada2024Heisenberg}. This problem was remedied by relying on the adaptive scheme, whose performance can be bounded in terms of the root mean square error (MSE)~\cite{wada2024Heisenberg}. 
However, it remains totally nontrivial whether one succeeds in efficiently obtaining collective properties in practical problem setups. A representative scenario is the partial tomography of fermionic systems that is useful for assessing quantum many-body simulation~\cite{yoshioka2022hunting, beverland2022assessing, meglio2024quantum}, estimating energy~\cite{babbush2018encoding, burg2021quantum, bauer2020quantum, mcardle2020quantum, lee2021even}, and probing the entanglement structure~\cite{cheong2004many, gullans2019entanglement, grover2013entanglement}.

\begin{figure*}
\hspace{-1cm}
\includegraphics[width=1.05\linewidth]{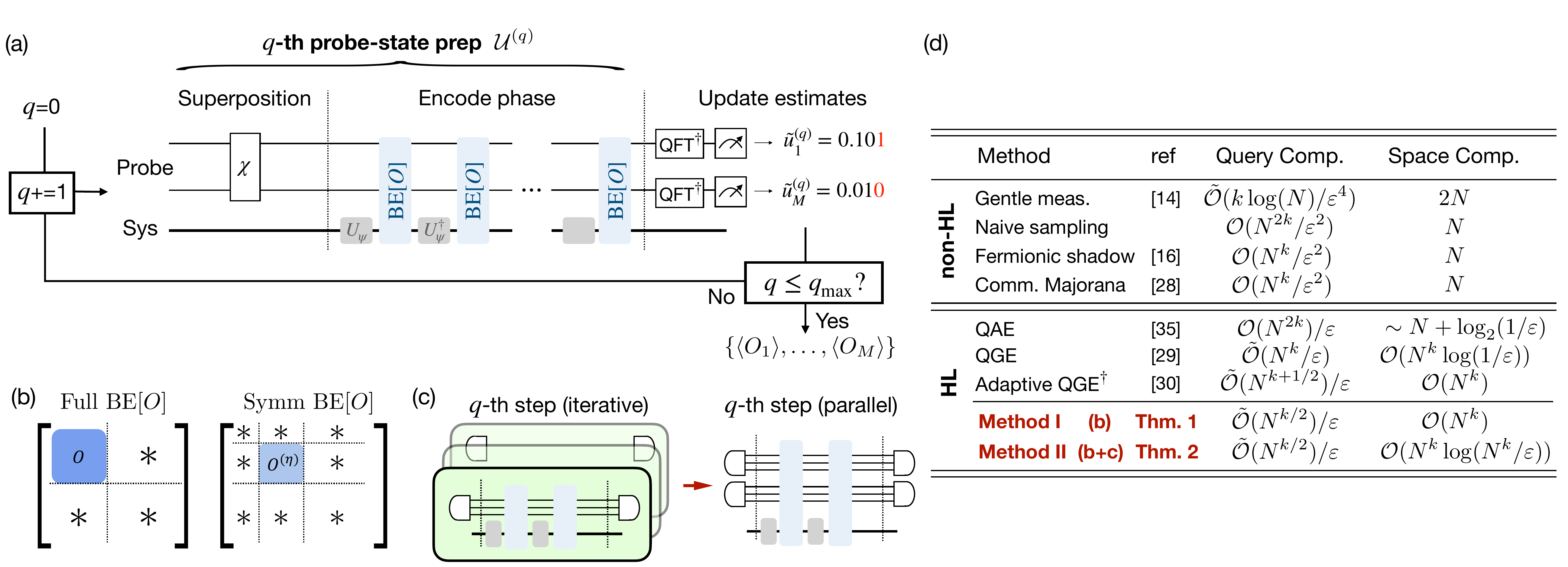}
\caption{
\justifying{
\label{fig:wide}
Schematic illustration of our proposal. (a) General structure of the adaptive QGE algorithm using oracular access to block encodings of $\{O_j\}_{j=1}^M$ and state preparation unitary $U_\psi^{(\dagger)}$. 
At $q$-th iteration, the subroutine $\mathcal{U}^{(q)}$ is called to encode the expectation values into the phases of the probe system, from which temporal estimates $\{\tilde{u}_j^{(q)}\}$ is obtained via measurement. The iteration is continued until $q_{\rm max} = \lceil \log_2(1/\varepsilon) \rceil$th step where desired precision $\varepsilon$ is reached.
(b) Block encoding of observables with and without utilizing the symmetry in the target system, exploited in Method I (see Theorem~\ref{thm:evaluation_symmetry_QGE_Informal}).
(c) Replacing iterative estimation with single-shot parallel readout. 
(d) Comparison of cost required to estimate all the elements of $N$-mode fermionic $k$-RDM with root MSE of $\varepsilon$ with a fixed particle number $\eta=k+\mathcal{O}(1)$ or $\eta=N-\mathcal{O}(1)$~\cite{huang2021information, Zhao:2020vxp, bonet2020nearly, huggins2022nearly, wada2024Heisenberg}. Here, Method II indicates the unified scheme for Method I and the parallel scheme (See Theorem~\ref{thm:evaluation_symmetry_parallel_QGE_main}). The dagger ($^\dagger$) indicates that the query complexity can be reduced quadratically if block-encodings of subspace-restricted operators $O_j^{(\lambda)} = \Pi_\lambda O_j \Pi_\lambda$ are given for any $j$, where $\Pi_\lambda$ is a projection operator onto the subspace labeled by $\lambda$.
}
}
\end{figure*}

In this work, we advance the QGE algorithm by presenting two novel variants, with the aim of tackling large-scale fermionic systems using the fewest queries among existing quantum algorithms (also see Fig.~\ref{fig:wide} for a summary).
The first variant dubbed Method I  exploits the symmetric structure in the target state, a ubiquitous property in quantum many-body problems, and the second variant, Method II, successfully performs single-shot adaptive estimation by utilizing the parallel scheme. 
The advantage of our proposals is demonstrated in fermionic systems with particle-number symmetry, for which Method II achieves the smallest query complexity under challenging systems such as the FeMo cofactor (FeMoco)  model or doped Fermi-Hubbard models. We concretely find that, when the target MSE for each element of 2-RDM is $\varepsilon=10^{-3}$, the total query count to the state preparation oracle is reduced by a factor of 100 for the FeMoco  model and by a factor of 500 for doped Fermi-Hubbard model on $100$ sites.



\textit{Problem Setup.---}
Let $|\psi\rangle$ be an $N$-qubit target state of interest, and
$\{O_j\}_{j=1}^M$ be a set of $2^N$-dimensional Hermitian observables satisfying $\|O_j\| \leq 1$ where $\norm{\cdot}$ denotes the spectral norm. We assume oracular access to the state preparation oracle
$U_\psi: \ket{0}^{\otimes N} \mapsto \ket{\psi}$ and its inverse, as well as block-encodings of observables $\{O_j\}$. 


Our goal is to estimate $\langle O_j \rangle \coloneq \bra{\psi} O_j \ket{\psi}$ within a root mean squared error $\varepsilon$, which can bound other error metrics such as the additive error~\cite{2009NJPh...11g3023H, PhysRevA.63.053804}. 
Since $U_\psi$ typically scales with the system size, the efficiency of the quantum algorithm is dominated by the total number of queries to $U_\psi$ and $U_\psi^\dagger$. As such, we evaluate the performance of quantum algorithms in terms of query complexity of $U_\psi$ and its inverse, informally summarized as follows.
\begin{problem}(Multiple Observables Estimation.) \label{prob:letter}
How do we construct estimators \(\{\hat{u}_j\}_{j=1}^M\) that satisfy, for all $j$,
\begin{align}
    \quad \mtr{MSE}[\hat{u}_j] = \mathbb{E} \left[ \left(\hat{u}_j - \langle \psi | O_j | \psi \rangle \right)^2 \right] \le \varepsilon^2,
\end{align}
using as few queries to \(U_\psi\) and \(U_\psi^\dagger\) for a target precision $\varepsilon \in (0, 1)$?
\end{problem}

\textit{General framework of adaptive QGE algorithm.---}
In similar to the phase estimation algorithm, the QGE algorithm first encodes the information of observable into phases of ancillary registers, or the {\it probe system}, and then read outs the values as binary information.
Adaptivity of the QGE algorithm resides in the state preparation subroutine of the probe system;  probe state at the $q$-th step is designed to estimate $q$-th digit of the expectation values.
Concretely, assume that one has temporary estimates of expectation values $\{\tilde{u}_j^{(q)}\}_{j=1}^M$ so that modified observables $A_j^{(q)} \coloneqq O_j - \tilde{u}_j^{(q)}$ satisfy 
\(
  |\langle \psi | A_j^{(q)} | \psi \rangle| \leq 2^{-q}.
\)
Then, to solve Problem~\ref{prob:letter}, it suffices to construct a subroutine $\mathcal{U}_\Upsilon^{(q)}$ that prepares copies of a probe state $\ket{\Upsilon(q)}$ that approximates the following ideal probe state with success probability at least $1-\delta^{(q)}$:
\begin{eqnarray}
  \ket{\Upsilon(q)} &\simeq& \sum_{\bm{x} } c_{\bm{x}} \,
  e^{i \phi_{\bm x}^{(q)}} \ket{\bm x}, 
  \phi_{\bm x}^{(q)} \propto 2^q\sum_{j=1}^M x_j  \langle \psi | A_j^{(q)} | \psi \rangle, \nonumber
\end{eqnarray}
where $\{c_{\bm x}\}$ is a fixed initial amplitude of the probe system.
Given such a probe state, inverse quantum Fourier transform and measurement allow us to obtain the updated temporary estimates $\{\tilde{u}_j^{(q+1)}\}$, in which we additionally have $(q+1)$-th bit information of $\{\langle O_j\rangle \}$.

We can easily see what determines the cost 
in this adaptive framework---probe-state preparation subroutine \(\mac{U}_\Upsilon\). Specifically, the total number of queries is given by
\begin{align}
    \text{(Total number of queries)} = \sum_{q=0}^{q_{\max}} \qty(\text{Cost of } \mac{U}_\Upsilon^{(q)}) .
\end{align}
We find that, to attain HL scaling, it suffices to ensure that $\qty(\text{Cost of } \mac{U}_\Upsilon^{(q)})$ is $ \mac{O} \qty(\aleph\cdot 2^q \log(1/\delta^{(q)}) ) $ for all $q$, where $\aleph$ is a positive factor independent of $q$ and $\varepsilon$. Indeed, the total query complexity is bounded if $\delta^{(q)}=c/8^{q_{\rm max} - q}$ where $c \in (0, 3/(8(1+\pi)^2)]$ as
\begin{align}
    \sum_{q=0}^{q_{\max}} \qty(\text{Cost of } \mac{U}_\Upsilon^{(q)}) &\leq 
     \aleph \sum_{q=0}^{q_{\max}} 2^q \log(1/\delta{(q)
     })  \nonumber\\
     &\leq\aleph \cdot 2^{q_{\max +1}} \log(8/c) =   \mac{O} (\aleph)/\varepsilon. \nonumber
\end{align}
This indicates that it is essential to reduce the prefactor \(\aleph\) without affecting its dependence on \(q\) and \(\varepsilon\).
Motivated by this requirement, we present strategies for reducing $\aleph$.

{{\it Method I: Observables estimation under symmetry.---}
One crucial property ubiquitous among numerous quantum many-body systems is the presence of symmetry in the target state $|\psi\rangle$. 
This motivates us to consider the decomposition of target observables into direct-sum form as
\begin{eqnarray}
    O_j = \bigoplus_{\lambda} O_j^{(\lambda)},\hfill
\end{eqnarray}
where \(O_j^{(\lambda)}\) is the component acting on the subspace labeled by $\lambda$.
The central idea of Method I is that, the complexity of the QGE algorithm should be determined only by the structure of symmetric subspace of the interest rather than the entire Hilbert space.

By building upon a symmetry-tailored framework of quantum singular value transformation (see Appendix~\ref{app:symmetric-qsvt}), we can prove the following theorem that explicitly quantifies the query complexity of our method.


\begin{theorem}(Observables estimation under symmetry.)
\label{thm:evaluation_symmetry_QGE_Informal}
Assume that target state $|\psi\rangle$ is supported on symmetric subspace labeled by $\lambda$.
Then, there exists a quantum algorithm that outputs a sample from estimators \(\{\hat{u}_j\}_{j=1}^M\) for \(\{\ \langle O_j \rangle\} \) satisfying
\begin{align}
\label{eq:MSE_symmetric_QSVT}
\max_{j=1,2,\ldots,M} \operatorname{MSE}[\hat{u}_j] \le \varepsilon^2,
\end{align}
using
\begin{align}
\label{eq:eval_query_sym_QGE_main}
\mathcal{O}\left( \varepsilon^{-1} \sqrt{ \norm{ \sum_{j=1}^M \left[O^{(\lambda)}_j\right]^2 } \log d_\lambda} \log M \right)
\end{align}
queries to the state-preparation oracles \(U_\psi\) and \(U_\psi^\dagger\) in total, where
$
d_\lambda
$
is the dimension of the symmetric subspace labeled by $\lambda$.
\end{theorem}

Importantly, our proposed algorithm does not require direct access to the projection operator \(\Pi_\lambda\) or block encoding of $\{O_j^{(\lambda)}\}$; rather, the query complexity is reduced solely by leveraging prior knowledge on the presence of symmetry and $\lambda$.

\noindent
\textit{Sketch of Proof.} 
To intuitively grasp how the speed up is achieved, it is informative to overview the ``encode phase" part in Fig.~\ref{fig:wide}(a). The encoding consists of three steps:
(i) constructing a block-encoding for $M^{-1}\sum_{j=1}^M x_j O_j$ controlled by $\bm x$, (ii) amplifying the normalization factor of this block-encoding via uniform singular value amplification \cite{gilyen2019quantum, low2017hamiltonian}, which results in a block-encoding of $\sigma^{-1}\sum_{j=1}^M x_j O_j$ for constant fraction of $\bm x$, where $\sigma = \mathcal{O}(\|\sum_j O_j^2\| \log d)$,
(iii) use of quantum singular value transformation for the optimal Hamiltonian simulation~\cite{gilyen2019quantum}, sandwitching $U_\psi$ and its inverse inbetween block encodings.

Observe that  only the step (iii) requires calls to \(U_\psi\) and its inverse. For a nonzero \(t \in \mathbb{R}\), optimal Hamiltonian simulation \cite{low2019hamiltonian} requires \(\mathcal{O}(t) \) calls to a block-encoding of the Hamiltonian \(H\) to implement \(e^{iHt}\). Hence, constructing the oracle embedding multiple target phases demands \(\mathcal{O}(\sigma)\) calls to the block-encoding of \(\sigma^{-1} \sum_{j} x_j \langle O_j \rangle\), which implies that the number of queries to \(U_\psi\) (and its inverse) scale in the same way.

If all $O_j$ share this direct-sum form and the target state is restricted to a fixed symmetric subspace labeled by $\lambda$, then at stage (ii) it suffices to implement the block-encoding of $\sigma_\lambda^{-1}\sum_j x_j O_j^{(\lambda)}$. By Lemma~\ref{lem:informal_subspace_QSVT} in Appendix~\ref{app:symmetric-qsvt} and uniform singular value amplification~\cite{apeldoorn2023quantum, wada2024Heisenberg}, one can construct a block-encoding of $\sigma_\lambda^{-1}\sum_j x_j O_j^{(\lambda)}$ that is valid for nearly all $\bm x$, where
\[
\sigma_\lambda = \mathcal{O} \qty(\sqrt{\norm{\sum_j \left(O_j^{(\lambda)}\right)^2} \log d_\lambda} ),
\]
and \(d_\lambda\) denotes the dimension of the symmetric subspace labeled by $\lambda$. By employing this block-encoding in stage (iii) and analyzing the overall query complexity, we obtain the upper bound expressed in Eq.~\eqref{eq:eval_query_sym_QGE_main}
(See Sec.V.B in the accompanying paper for the complete proof).
\hfill \hfill \qed

\black{{\it Method II: Further use of parallel scheme.---}}
Next, we turn to another key ingredient for achieving speedup: parallel scheme. 
We note that our implementation does not employ conventional parallel scheme—in which many identical quantum circuits are executed concurrently—as such an approach does not reduce the total query complexity. Rather, our proposal enables a single-shot measurement at each iteration of adaptive scheme. Namely, instead of iteratively reading out from $\ket{\Upsilon(q)}$ for $R^{(q)}$ times to take the median to enhance success probability, we prepare enlarged entangled state $\bigotimes_{r=1}^R \ket{\Upsilon(q)}$ and extract the estimates simultaneously, as illustrated in Fig.~\ref{fig:wide}(c). Such an entanglement-assisted readout reduces the query complexity from $\mathcal{O}(\varepsilon^{-1} R\sqrt{M \log d})$ to $\mathcal{O}(\varepsilon^{-1} \sqrt{MR \log d})$, achieving quadratic speedup regarding $R$. 
Since it suffices to take $R^{(q)}=\mathcal{O}(\log(M))$ for each $q$, the factor $\aleph$ characterizing the complexity of the subroutine \(\mathcal{U}^{(q)}_{\Upsilon}\) becomes \(\aleph = \mathcal{O}(\sqrt{M \log M \log d})\).
By integrating the two novel proposals, we can develop an even more efficient QGE algorithm. The following theorem provides an explicit asymptotic evaluation, 
\begin{theorem}(Observables estimation under symmetry and single-shot parallel scheme.)
\label{thm:evaluation_symmetry_parallel_QGE_main}
Assume that target state $|\psi\rangle$ is supported on symmetric subspace labeled by $\lambda$.
Then, there exists a quantum algorithm that outputs a sample from estimators \(\{\hat{u}_j\}_{j=1}^M\) for \(\{\ \langle O_j \rangle\} \) satisfying
\begin{align}
\label{eq:MSE_symmetric_QSVT}
\max_{j=1,2,\ldots,M} \operatorname{MSE}[\hat{u}_j] \le \varepsilon^2,
\end{align}
using
\begin{align}
\label{eq:eval_query_sym_para_QGE_main}
\varepsilon^{-1}  \cdot \mathcal{O}\left( \sqrt{ \norm{ \sum_{j=1}^M \left[O^{(\lambda)}_j\right]^2 } \log d_\eta \log M} \right)
\end{align}
queries to the state-preparation oracles \(U_\psi\) and \(U_\psi^\dagger\) in total, where
$
d_\lambda
$
is the dimension of the symmetric subspace labeled by $\lambda$.
\end{theorem}
\noindent

For the complete proof of the Theorem, we guide the readers to Sec.V.C of accompanying paper~\cite{koizumi2025full}.


\black{\textit{Application to fermionic problems.---}}
When investigating electron correlation in fermionic systems, the key observables are the fermionic reduced density matrices (RDMs) \cite{lowdin1955quantum}. For integer $k \geq 1$, the elements of the $k$-RDM (${}^kD$) are given by
\begin{align}
    {}^k D^{\bm p}_{\bm q} \coloneq \bra{\psi} a_{p_1}^\dagger \cdots a_{p_k}^\dagger\, a_{q_1} \cdots a_{q_k} \ket{\psi},
\end{align}
where $p_j$ and $q_j$ label the $N$ fermionic modes, and $a_p^\dagger$, $a_q$ are creation and annihilation operators. 
Obviously, each operator yielding $k$-RDM element conserves the particle number, and thus Hermitianized operators such as $(^k D_{\bm p}^{\bm q} + {}^kD_{\bm q}^{\bm q})/2$ admit direct-sum representation labeled by $\eta$. We find that 
$\bigl\|\!\sum_j [O_j^{(\eta)}]^2 \bigr\| = \binom{\eta}{k} \binom{N-\eta+k}{k},$
which significantly reduces the query complexity in Theorems~\ref{thm:evaluation_symmetry_QGE_Informal} and \ref{thm:evaluation_symmetry_parallel_QGE_main}.


We assess the asymptotic cost of various strategies for estimating $k$-RDM elements to a root mean square error $\varepsilon$, with key findings illustrated in Fig.~\ref{fig:wide}(d). Here, both the real and imaginary parts of each $k$-RDM element are required to meet precision $\varepsilon$. Under the extreme condition $\eta = k + \mathcal{O}(1)$ or $\eta = N - \mathcal{O}(1)$, our proposals offer an additional speedup. Although Bell and gentle measurements~\cite{huang2021information}---like shadow tomography~\cite{aaronson2018shadow}---can exponentially reduce measurement overhead at the cost of precision, accurately capturing strongly correlated systems still requires high-precision estimates of $k$-RDM elements, granting our approach an asymptotic advantage when $\varepsilon \in o(N^{-k/6})$. Moreover, our method surpasses classical shadow algorithms and achieves a quartic speedup over the QAE algorithm, establishing the best asymptotic performance in the high-precision regime.


\begin{figure}[t]
    \centering
\includegraphics[width=0.9\linewidth]{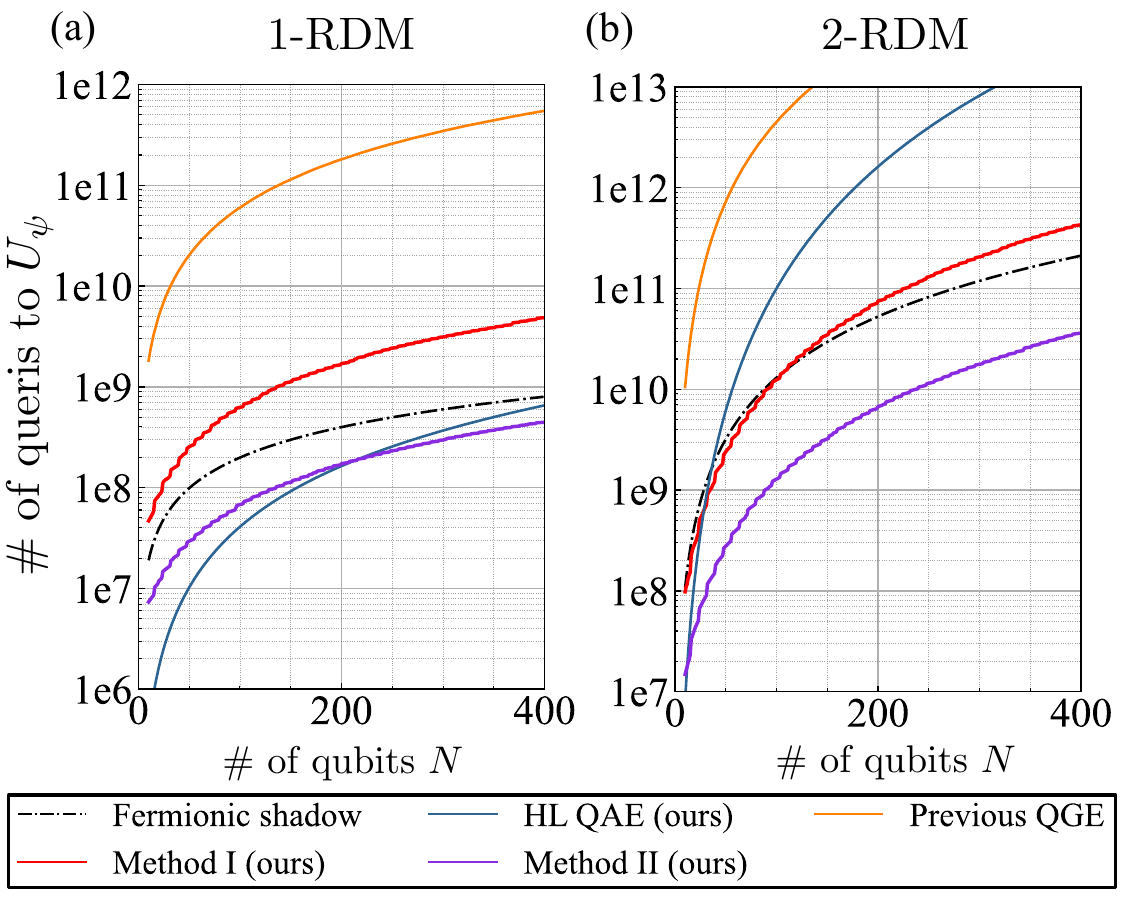}
    \caption{
    \justifying{The total query complexity in (a) $1$-RDM elements and (b) $2$-RDM elements estimation for $N$-qubit fermionic systems with $\eta = \lceil 7N/8  \rceil$ particles. Here, the target precision is set as $\varepsilon = 10^{-3}$. The approaches compared include, the Fermionic shadow tomography \cite{Zhao:2020vxp}, the quantum amplitude estimation (QAE) algorithm with HL scaling, the previous adaptive QGE algorithm \cite{wada2024Heisenberg}, and our proposals.  }
    }
    \label{fig:RDMvsN}
\end{figure}

\begin{figure}[t]
    \centering
\includegraphics[width=0.9\linewidth]{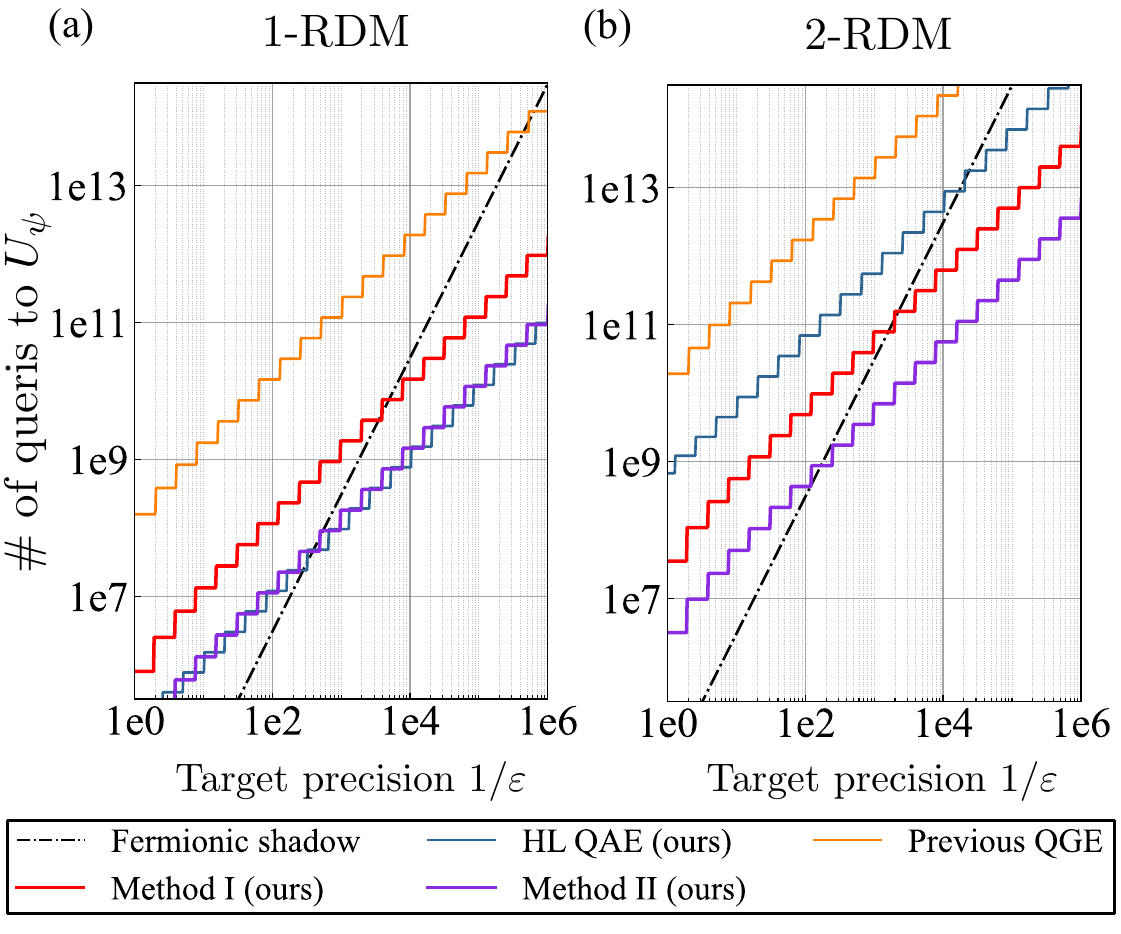}
    \caption{
    \justifying{The total query complexity in (a) $1$-RDM elements, (b) $2$-RDM elements, when varying a root mean squared error $\varepsilon$. Here, we consider the $k$-RDM elements estimation for $N = 152$ active space of FeMoco with $\eta = 113$ electrons \cite{li2019electronic}. The algorithms compared here are the same as those in Fig.~\ref{fig:RDMvsN}. }
    }
    \label{fig:RDMvserr}
\end{figure}

It is crucial to note, however, that asymptotic evaluations alone do not fully capture the practical performance of these algorithms; constant and logarithmic factors, as well as circuit overhead and implementation details, can significantly influence the total query count in real-world applications.
To complement our analysis, we also perform numerical evaluations of the query complexities for various algorithms in practical scenarios. 

Figure~\ref{fig:RDMvsN} shows the query count to estimate 1,2-RDM elements with accuracy of $\varepsilon=10^{-3}$ for fermionic systems with filling of $\eta = \lceil 7N/8 \rceil$. Such a setup reflects the unresolved phase diagram of doped Fermi-Hubbard model on 2D lattice. We find that our proposals excel at wide regimes; for 1-RDM estimation, the QAE algorithm proposed in accompanying paper~\cite{koizumi2025full} and Method II achieves the lowest number of queries for $N\leq 200$ and $N \geq 200$, respectively, and for 2-RDM estimation, the Method II is superior for any $N$.
Furthermore, for $k \ge 3$, since the asymptotic query complexity scales as \(\tilde{\mathcal{O}}\qty( \mqty(\eta \\ k) \mqty(N- \eta+k  \\ k)  )/\varepsilon\), the gap between our proposal and other methods becomes even more pronounced (see accompanying paper~\cite{koizumi2025full}). 

Scaling with the target accuracy $\varepsilon$ is also of great interest to practitioners. With the active space model of nitrogenase FeMo cofactor in our mind, we evaluate the query count for $N=152$ modes and $\eta=113$ particles, as shown in Fig.~\ref{fig:RDMvserr}. For the target precision of $\varepsilon \lesssim 10^{-2}$, we find that our proposal yields the lowest query count among existing methods.
In practice, the precision needed to resolve strong correlation effects often matches chemical accuracy (e.g., $\varepsilon = 10^{-3}$); systematic studies by Tilly \textit{et al.}\ show that statistical errors in individual 2-RDM elements must be reduced to $\varepsilon \lesssim 10^{-3}$ to meet the accuracy~\cite{tilly2021reduced}.

\textit{Summary and outlook.---}
In this work, we have proposed two novel variants of the quantum gradient estimation (QGE) algorithm: one harnesses the intrinsic symmetries of quantum systems, and the other further employs parallel schemeat the expense of extra qubits. To our knowledge, these algorithms are the first to achieve quartic reduction of the cost compared to the quantum amplitude estimation algorithm for fermionic system with specific particle numbers (see Fig.~\ref{fig:wide}(d)).
Our approach provides not only superior asymptotic performance but also practical advantages in estimating the expectation values of \(k\)-local fermionic operators. Because these operators are key to understanding electron correlation in materials and the complexity of atomic structure, our proposal is expected to be valuable across a broad spectrum of fields, including condensed matter physics, quantum chemistry, and high-energy physics.

Among numerous intriguing future directions, we highlight here the two most crucial ones.
First, it is practically crucial to investigate how error in algorithm and hardware noise affect the performance of the algorithm. It is nontrivial whether calibration technique for phase estimation can be straightforwardly applied to protect the measurement result~\cite{kimmel2015robust}.
Second, given that our algorithm relies on Hamiltonian simulation to encode the information of observables into phases, it is theoretically interesting and nontrivial to ask whether the algorithm can leverage fast-forwarding approaches, which can even compress the cost exponentially for commuting Hamiltonians~\cite{atia2017fast}.

\textit{Acknowledgements.---}
The authors wish to thank Dominic Berry, Yosuke Mitsuhashi, Takahiro Sagawa, and Kento Tsubouchi for fruitful discussions.
Y. K. is supported by the Program for Leading Graduate Schools (MERIT-WINGS).
K. W. was supported by JSPS KAKENHI Grant Number JP24KJ1963.
W. M. is supported by MEXT Quantum Leap Flagship Program
(MEXTQLEAP) Grant No. JPMXS0120319794,  the
JST COI-NEXT Program Grant No. JPMJPF2014,
the JST ASPIRE Program Grant No. JPMJAP2319,
and the JSPS
Grants-in-Aid for Scientific Research (KAKENHI) Grant
No. JP23H03819.
N.Y. is supported by JST Grant Number JPMJPF2221, JST CREST Grant Number JPMJCR23I4, IBM Quantum, JST ASPIRE Grant Number JPMJAP2316, JST ERATO Grant Number JPMJER2302, and Institute of AI and Beyond of the University of Tokyo.

\clearpage
\appendix

\section{Subspace quantum singular value transform} \label{app:symmetric-qsvt}

In the framework of quantum singular value transformation (QSVT)~\cite{gilyen2019quantum}, a block-encoding of $f(O)$ can be constructed using the block-encoding $B$ of a Hermitian operator $O$ and a sequence of phase gates $\{e^{i\phi_k Z}\}$, where $f$ is a polynomial function that acts on the singular values of $O$.
Assume that $O$ has a direct-sum structure $O = \bigoplus_{\lambda} O^{(\lambda)}$ where $\lambda$ labels each subspace, and that we are interested in a specific $\lambda.$
Notably, in such a case, we only need to approximate $f$ over the singular values of $O^{(\lambda)}$ to implement $f(O^{(\lambda)}).$ 
This observation underpins the QSVT only effective to the specific subspace, formally stated in the following lemma.
\begin{lemma}
    \label{lem:informal_subspace_QSVT}
    Assume that $O$ is a Hermitian operator with a direct-sum structure $O = \bigoplus_\lambda O^{(\lambda)}$, where each $O^{(\lambda)}$ acts on a subspace labeled by $\lambda$. If we have a block-encoding of $O$, then applying quantum singular value transformation circuit with a polynomial function $f$ results in a block-encoding of the transformed operator $\bigoplus_\lambda f(O^{(\lambda)})$. 
\end{lemma}

\begin{proof}
From the direct sum decomposition, we can derive the eigenvalue decomposition of each $d_\lambda$-dimensional operator $O^{(\lambda)}$ and denote it as $O^{(\lambda)} = U^{(\lambda)} \Sigma^{(\lambda)} (U^{(\lambda)})^\dagger$ where $U^{(\lambda)}$ is a $d_\lambda \times d_\lambda$ unitary matrix. Since the $O^{(\lambda)}$ is Hermitian, there exists an eigenvalue decomposition $O^{(\lambda)} = U^{(\lambda)} \Sigma^{(\lambda)} (U^{(\lambda)})^\dagger$ where $U^{(\lambda)}$ is a unitary matrix and $\Sigma^{(\lambda)}$ is a diagonal matrix. This eigenvalue decomposition leads to 
\begin{align}
     O 
     = \bigoplus_{\lambda \in \Lambda} U^{(\lambda)}  \Sigma^{(\lambda) \dagger} U^{(\lambda)}.
\end{align}
Here, we define $U \coloneq \bigoplus_{\lambda \in \Lambda} U^{(\lambda)} $, and it is trivial to show that $U$ is unitary. Since $\bigoplus_{\lambda \in \Lambda} \Sigma^{(\lambda)} $ is a diagonal matrix, we can consider the eigenvalue decomposition of $O$ is equivalent to $U \bigoplus_{\lambda \in \Lambda} \Sigma^{(\lambda)}   (U)^\dagger $.

Now let us assume that we know a phase sequence $\{\phi_k\}$ for a degree-$m$ polynomial $f$.
Given a $a$-block-encoding $U_O$ of a Hermitian operator $O$, we can implement
\begin{align}
    f^{\mtr{(SV)}} (O) \coloneq U f \qty(  \bigoplus_{\lambda \in \Lambda } \Sigma^{(\lambda)}   ) U^\dagger.
\end{align}
Since $f$ is a polynomial, we can demonstrate that
\begin{align}
    U f \qty(  \bigoplus_{\lambda \in \Lambda } \Sigma^{(\lambda)}) U^\dagger &= f  \qty( U  \qty( \bigoplus_{\lambda \in \Lambda }  \Sigma^{(\lambda)} )  U^\dagger ) \\
    &= f \qty(\bigoplus_{\lambda \in \Lambda} U^{(\lambda)} \Sigma^{(\lambda)} (U^{(\lambda)})^\dagger )  \\
    &= f \qty( \bigoplus_{\lambda \in \Lambda} O^{(\lambda)} ) = \bigoplus_{\lambda \in \Lambda } f(O^{(\lambda)} ).
\end{align}
\end{proof}

\section{Pseudocode for general framework of adaptive QGE}

For the sake of completeness, here we provide the pseudocode to describe the general framework of adaptive QGE algorithm.

\begin{algorithm}[H]
    \caption{General framework of adaptive QGE algorithm for multiple observables}
  \label{alg:unified_framework}
  \begin{algorithmic}[1]
  \Statex \textbf{Input:} $\log_2 d$-qubit state preparation unitary $U_\psi$ and its inverse, $M$ observables $\{O_j\}_{j=1}^M$ of bounded spectral norm $\norm{O_j} \leq 1$ where  $M \geq 2\log_2 d + 24 $, 
  confidence parameter $c \in (0, \frac{3}{8(1 + \pi)^2}]$, target precision parameter $\varepsilon \in (0, 1)$, an integer $p\geq 1$, and a set of integers $\{R^{(q)}\}_{q=0}^{\lceil \log_2(1/\varepsilon) \rceil }$
  \smallskip

  \Statex \textbf{Subroutine:} a probe-state preparation subroutine $\mathcal{U}_{\Upsilon}^{(q)}(\{A_j\}_{j=1}^M)$.
  This process is ensured to work as follows:
  for given integer $q\geq 0$ and observables $\{A_j\}_{j=1}^M$ of bounded spectral norm $\norm{A_j} \leq 1$
  , the process $\mathcal{U}_{\Upsilon}^{(q)}$ prepares $R^{(q)}$ copies of a $pM$-qubit quantum state \(\ket{\Upsilon(q)}\) that is close in the Euclidean norm to the following ideal probe state under fixed set of amplitudes $\{c_{\bm x}\}$:
  \begin{eqnarray}
      \ket{\Upsilon(q) } \simeq \sum_{\bm x \in G_p^M} c_{\bm x} 
      e^{2\pi i {2^p}{\sum_{j=1}^M x_j 2^{q} \pi^{-1} \bra{\psi}A_j\ket{\psi}}} \ket{\bm x},\nonumber\\ \mbox{if}~~~ |\bra{\psi}A_j\ket{\psi}|\leq 2^{-q},~\forall j.
  \end{eqnarray}

  \Statex \textbf{Output:} A sample from an estimator $\hat{u} = (\hat{u}_1, \dots, \hat{u}_M)$ whose $j$-th element estimates $\langle O_j \rangle := \langle \psi | O_j | \psi \rangle$ within MSE $\epsilon^2$ as 
  $$
  \max_{j=1,2,\dots,M} \mathbb{E}[(\hat{u}_j - \langle O_j \rangle)^2] \leq \epsilon^2 
  $$
  \State $\tilde{u}_j^{(0)} \leftarrow 0$ for $j = 1, 2, \dots, M$
  \For{$q = 0, 1, \dots, q_{\max} \coloneq \lceil \log_2(1/\epsilon) \rceil$}
      \State $A_j \leftarrow {O}_j-\tilde{u}^{(q)}_j\bm{1}$
      \State Call the subroutine $\mathcal{U}_{\Upsilon}(q,\{A_j\})$ for preparing $R^{(q)}$ copies of the quantum state $\ket{\Upsilon(q)}$
      \State Apply  $(\mtr{QFT}_{G_p}^\dagger)^{\otimes M}$ on each copy 
      \State Perform computational basis measurement to obtain output
      $(k_1, \ldots, k_M) \in G_p^M$.  
      \State $g_j^{(q)} \leftarrow$ coordinate-wise medians of the measurement results \label{step:get-coordinate-wise-median}
      \State $\tilde{u}_j^{(q+1)} \leftarrow \tilde{u}_j^{(q)} + {\pi}{2^{-q}} g_j^{(q)}$.
      \For{$j=1, ..., M$}
      \If{$\tilde{u}_j^{(q+1)}\geq 1$ (or $\leq -1$)} 
          \State $\tilde{u}_j^{(q+1)} \leftarrow 1$ (or $-1$) 
      \EndIf
      \EndFor
  \EndFor
  \end{algorithmic}
\end{algorithm}

\bibliography{ref} 

\end{document}